\documentclass[a4paper,10pt,preprintnumbers,floatfix,superscriptaddress,aps,pra,twocolumn,showpacs,notitlepage,longbibliography]{revtex4-2}
\usepackage[utf8]{inputenc}
\usepackage[T1]{fontenc}
\usepackage[sc,osf]{mathpazo}
\usepackage{amsmath,amsthm,mathtools}
\usepackage{bbm}
\usepackage{marvosym}
\usepackage{graphicx}
\usepackage{fancyvrb}
\usepackage{tikz}
\usetikzlibrary{matrix}
\usepackage{xcolor}
\usepackage{xr}
\usepackage{float}
\usepackage{appendix}
\usepackage{hyperref}
\usepackage{zref-clever}
\usepackage[capitalise]{cleveref}
\usepackage[caption=false]{subfig}

\renewcommand{\Pr}{\mathbb{P}}
\newcommand{\eps}{\varepsilon}
\newcommand{\epsi}{\eps_{\mathrm{I}}}
\newcommand{\epsii}{\eps_{\mathrm{II}}}
\DeclareMathOperator{\LandauO}{\mathrm{O}} 

\providecommand\given{}
\newcommand\SetSymbol[1][]{%
  \nonscript\:#1\vert
  \allowbreak
  \nonscript\:
  \mathopen{}}
\DeclarePairedDelimiterX\Set[1]\{\}{%
  \renewcommand\given{\SetSymbol[\delimsize]}
  #1
}
\DeclarePairedDelimiterX{\abs}[1]{\lvert}{\rvert}{{#1}} 

\DeclareMathOperator{\tr}{tr}

\newcommand{\CC}{C^\ast}
\newcommand{\IP}{\mathrm{IP}}
\newcommand{\DISJ}{\mathrm{DISJ}}

\newcommand{\EQ}{\mathrm{EQ}}
\newcommand{\INDEX}{\mathrm{INDEX}}
\newcommand{\kINT}{k\text{-}\mathrm{INT}}

\usepackage{thmtools}
\usepackage{thm-restate}
\newtheorem{theorem}{Theorem}

\newtheorem{proposition}[theorem]{Proposition}
\newtheorem{observation}[theorem]{Observation}
\newtheorem{example}{Example}

\usepackage[nolist,withpage]{acronym}
\newacro{CHSH}{Clauser-Horne-Shimony-Holt}
\newacro{PR}{Popescu-Rohrlich}
\newacro{IC}{information causality}

\begin{document}
\title{Communication complexity bounds from information causality}
\author{Nikolai Miklin}
\thanks{These authors contributed equally to this work.\\ \{nikolai.miklin,prabhav.jain\}@tu-darmstadt.de}
\affiliation{Institute for Quantum-Inspired and Quantum Optimization, Hamburg University of Technology, Germany}
\affiliation{Institute for Applied Physics, Technical University of Darmstadt, Darmstadt, Germany}
\author{Prabhav Jain}
\thanks{These authors contributed equally to this work.\\ \{nikolai.miklin,prabhav.jain\}@tu-darmstadt.de}
\affiliation{Department of Computer Science, Technical University of Darmstadt, Germany}
\author{Mariami Gachechiladze}
\affiliation{Department of Computer Science, Technical University of Darmstadt, Germany}

\date{February 11, 2026}
\begin{abstract}
Communication complexity, which quantifies the minimum communication required for distributed computation, offers a natural setting for investigating the capabilities and limitations of quantum mechanics in information processing. We introduce an information-theoretic approach to study one-way communication complexity based solely on the axioms of mutual information. 
Within this framework, we derive an extended statement of the information causality principle, which recovers known lower bounds on the communication complexities for a range of functions in a simplified manner and leads to new results.
We further prove that the extended information causality principle is at least as strong as the principle of non-trivial communication complexity in bounding the strength of quantum correlations attainable in Bell experiments.
Our study establishes a new route for exploring the fundamental limits of quantum technologies from an information-theoretic viewpoint.
\end{abstract}
	
\maketitle
\hypersetup{pdftitle = {Communication complexity bounds from information causality},
           pdfauthor = {Nikolai Miklin, Prabhav Jain, Mariami Gachechiladze},
           pdfsubject = {Quantum foundations},
           pdfkeywords = { 
                  quantum, communication, complexity, information causality,
                  non-trivial communication complexity,
                  entanglement-assisted communication,
                  physical principle, Bell, nonlocality, non-locality,
                  axioms, mutual information, set of quantum correlations,
                  }
      }
	
\section{Introduction}
Quantum information technology is a rapidly developing field, motivated by the potential advantages that quantum mechanics can offer for information-processing tasks, particularly in communication and computation. 
Studying the limitations of these quantum advantages is as important as exploring new applications, not only from a foundational perspective, but also for practical considerations.
It is especially interesting when these limitations match the best known algorithms, as is the case, for example, with Grover's search algorithm~\cite{grover1996fast,bennett1997strengths}.

At the intersection of communication and computation lies the task of \emph{distributed computing}, which provides a natural setting to investigate the capabilities and limitations of quantum mechanics in information processing. 
In this setting, two or more parties collaboratively compute a function whose input is distributed among them. The idealized model assumes that each party has unlimited computational power, and the objective is to minimize the amount of communication required. The minimum communication necessary to compute a function is called its \emph{communication complexity}. This concept arises naturally in distributed systems and has broad implications for theoretical computer science, including minimal models of computation and data structures~\cite{kushilevitz1996communication,hromkovic2010communication}.

The idea that quantum communication or shared entanglement could offer advantages in distributed computing was recognized early on~\cite{yao1993quantum,cleve1997substituting}. In fact, the savings in terms of communicated qubits, or classical bits, respectively, can be \emph{exponential}~\cite{buhrman1998quantum}, even when allowing for some error in computation~\cite{raz1999exponential}.
At the same time, quantum advantages in communication complexity are expected to be subject to fundamental limits. 
In fact, it is conjectured that for total functions, i.e., for the situations where there is no restriction on the input to the distributed computation, the quantum communication complexity is at most polynomially smaller than the classical counterpart (see e.g.,~\cite{jain2009}).

Of particular interest, motivated by the afore-mentioned long standing conjecture, and by applications in neighbouring areas of complexity theory, such as streaming algorithms~\cite{lovett2023}, are \emph{lower bounds} on quantum communication complexity~\cite{nayak1999,Buhrman2001,deWolf2000,klauck2001,sherstov2008,sherstov2010,linial2007,jain2003,vandam2002renyi,Montanaro2007lower,ambainis2006,aaronson2004,jain2009,Boddu2023relatingoneway,Razborov2003,podolskii2024,zhang2011,chakrabarti2003,ziv2004}.
Despite an extensive body of literature, only a few methods are known that can be applied to derive lower bounds for general functions~\cite{aaronson2004,linial2007}.
Somewhat surprisingly, the restriction on the number of rounds of the communication protocol, in particular, to the \emph{one-way} model, in which only one party sends a message to the other, does not appear to make the problem easier~\cite{aaronson2004}.

In this paper, we propose an extension of the information causality principle~\cite{pawlowski2009information} to distributed computing scenario. 
So far, this principle has been proposed and applied to bounding the strength of quantum correlations in Bell experiments~\cite{bell1964einstein,cirelson1980quantum,pawlowski2009information,allcock2009,miklin2021information,gachechiladze2022quantum,jain2024informationcausality}, relying solely on the set of axioms for mutual information.
Here, we show how an extension of this principle can be applied to derive lower bounds on the communication complexity for arbitrary functions in the bipartite scenario with one-way entanglement-assisted classical communication.
Via the teleportation protocol~\cite{bennett1985teleporting}, the same bound applies to one-way quantum communication, up to a factor of two, and also to public-coin classical communication complexity, since entanglement is a stronger resource.
We illustrate our approach with several examples. For some functions, we recover known asymptotically tight lower bounds using substantially simpler arguments, while for others, we obtain new results.

An additional motivation to study bounds on communication complexity in the entanglement-assisted classical communication model comes from their close connection to Bell nonlocality~\cite{buhrmann2010nonlocality}.
A central and counter-intuitive result in this direction is due to van Dam~\cite{vandam2013implausible,cleve2000unpublished}: if the parties have access to distributed correlations, known as the \ac{PR} boxes~\cite{popescu1994quantum}, then any function can be computed in a distributed manner using only a single bit of communication for any input length.
Motivated by this observation, it was later proposed as a physical principle, known as \emph{non-trivial communication complexity}, that there must exist functions whose communication complexity grows with the size of their input~\cite{brassard2006limit}.
We show that the proposed extension of the information causality principle implies the principle of non-trivial communication complexity and establishes a new tool for bounding the set of quantum correlations.

\section{Preliminaries}
In the standard communication complexity scenario, one considers two parties, Alice and Bob, who face the problem of computing the value of a Boolean function (known to them in advance) on some random input distributed between them in each round~\cite{yao1979some}.
The objective is to minimize the amount of communication between the parties required to compute the function successfully on every input. Various types of scenarios are distinguished depending on the type of communication (one-way or two-way), the resources available (e.g., randomness or shared entanglement), and the allowed degree of error in the computed value.
As one would expect, the communication complexity can vary significantly depending on the case considered. We refer the reader to Refs.~\cite{klauck2000quantumc,brassard2003quantum,buhrmann2010nonlocality,deWolf2022quantum} for a comprehensive, although slightly outdated, overview of the known results in quantum communication complexity.

In this work, we focus on the one-way entanglement-assisted model depicted in \cref{fig:scenario}. Let the input sets $X$ and $Y$ be general finite sets. 
It is most common to consider $X=Y=\{0,1\}^n$ for some positive integer $n$, but we also encounter other cases, e.g., $Y=\{0,1,\dots,n-1\}\eqqcolon [n]$. 
We only consider the case when Alice's and Bob's inputs, $x$ and $y$, can take any values from $X$ and $Y$, respectively, which is often referred to as the case of a \emph{total} function $f:X\times Y \to \{0,1\}$, as opposed to case when there is some promise on pairs $(x,y)$ that can occur. For our purposes, it is convenient to introduce a random variable $g$ taking values in $\{0,1\}$, which represents Bob's guess of the value of the function $f$ evaluated on $x$ and $y$. 
We restrict our analysis to the binary alphabet, as many of the relevant phenomena already manifest in this case. Nevertheless, the formalism introduced here can be naturally adapted to arbitrary finite alphabets.

We define the communication complexity of a function $f$ in the considered scenario in \cref{fig:scenario} as the fewest bits of communication needed, 
\begin{align}\label{eq:def_cc}
	\CC_\eps(f)\coloneqq \min~& m\\
    \text{s.t.}~ & \Pr[g=f(i,j)\vert x=i,y=j]\geq 1-\eps,\; \forall i,j,\nonumber
\end{align}
where $\eps\in [0,\frac12)$ is the allowed error in distributively computing $f$ for each pair of inputs $i\in X$, $j\in Y$, and the minimization is taken over all communication protocols.
Following common notation in the literature, we use an asterisk ($\ast$) to indicate that parties have access to unlimited shared entanglement.
The case of $\eps=0$ corresponds to the so-called deterministic communication complexity $\CC_0(f)$. 

Naturally, we are not interested in the exact values of $\CC_\eps(f)$ but in its asymptotic dependence on the cardinality of $X$ and $Y$. 
Importantly, the error bias $\eps$ in \cref{eq:def_cc} must be independent of the size of the input.
A trivial observation is that $\CC_\eps(f)\leq \CC_0(f)\leq \log(\abs{X})$, i.e., guessing with error requires less communication than guessing perfectly, and if Alice sends her input $x$ to Bob using at most $\log(\abs{X})$ bits, he can calculate the value of any function of $x$ and $y$ deterministically.

\begin{figure}[t!]
	\includegraphics[width=.9\linewidth]{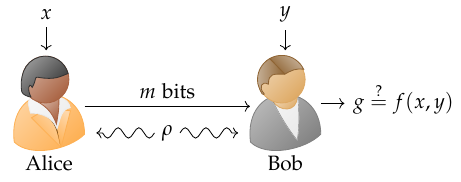}
	\caption{Considered scenario: at each round, Alice and Bob receive random inputs $x$ and $y$, taking values in $X$ and $Y$, respectively. Alice sends $m$ bits of data to Bob, who then produces a guess $g$ of the value of $f(x,y)$. Alice and Bob have access to unlimited shared entanglement.}
	\label{fig:scenario}
\end{figure}

In this work, we study the following commonly considered functions:
\begin{enumerate}
    \item[1.] the index function
    \begin{equation}\label{eq:def_index}
    	\INDEX_n(x,y)\coloneqq x_y\ ,
    \end{equation}
    for $X=\{0,1\}^n$, $Y=[n]$, which has been extensively studied in the context of random access codes~\cite{Ambainis1999};
    \item[2.] the inner product function
    \begin{equation}\label{eq:def_ip}
    	\IP_n(x,y)\coloneqq \bigoplus_{i=0}^{n-1} x_i\cdot y_i\ ,
    \end{equation}
    for $X=Y=\{0,1\}^n$, where $x_i$ and $y_i$ are the $i$-th bits of $x=x_0x_1\dots x_{n-1}$ and $y=y_0y_1\dots y_{n-1}$, respectively, and $\oplus$ is the sum modulo $2$;
    \item[3.] the disjointness function
    \begin{equation}\label{eq:def_disj}
    	\DISJ_n(x,y)\coloneqq \begin{cases} 
    		0 & \exists\, i: x_i=y_i=1,\\
    		1 & \text{otherwise},
    	\end{cases}
    \end{equation}
    also for $X=Y=\{0,1\}^n$; and
    \item[4.] the equality function
    \begin{equation}\label{eq:def_gt}
    	\EQ_n(x,y)\coloneqq [x=y]\ ,
    \end{equation}
    for $X=Y=[2^n]$, where $[\cdot]$ denotes the Iverson bracket.
\end{enumerate}
The entanglement-assisted one-way classical communication complexity of the above functions is known. 
For the inner product, disjointness, and index functions $\CC_\eps(\IP_n),\CC_\eps(\DISJ_n),\CC_\eps(\INDEX_n)\in\Theta(n)$~\cite{cleve1998ip,klauck_gt_disj_q,nayak1999}.
Surprisingly, for the equality function it only holds that $\CC_0(\EQ_n)\in\Omega(n)$~\cite{Hyer2002_q_eq_exact}, while  $\CC_\eps(\EQ_n)\in \LandauO(1)$ for $\eps>0$~\cite{kushilevitz1996communication}. 
Additionally, we consider
\begin{enumerate}
\item[5.] the $k$-intersect function
\begin{equation}\label{eq:def_kint}
	\kINT_n(x,y)\coloneqq \left[\sum_{i=0}^{n-1}x_i\cdot y_i \geq k\right]\ ,
\end{equation}
for $\quad X=Y=\{0,1\}^n$ and $k\in\Set{1,2,\dots,\lfloor n/2\rfloor}$. 
For $k=1$, this function reduces to the standard intersect function, for which it is known that $\CC_\eps(1\text{-}\mathrm{INT}_n)\in \Theta(n)$~\cite{klauck_gt_disj_q}, while the case of $k\geq 2$, to the best of our knowledge, has not been considered so far in the literature.
\end{enumerate}

\section{Results}
\label{sec:ULB}
We start by assuming a measure of conditional mutual information $I(A;B\vert C)$ defined for systems $A,B$, and $C$, that satisfies the following set of axioms:
\begin{enumerate}
	\item[(i)] non-negativity: $I(A;B\vert C)\geq 0$,
	\item[(ii)] chain rule: $I(A;B\vert C) = I(A;B,C)-I(A;C)$,
	\item[(iii)] data processing: $I(A;B\vert C) \geq I(\Phi(A);B,C)$, for a physical map $\Phi$,
	\item[(iv)] and consistency with Shannon mutual information in case $A,B$, and $C$ are (classical) random variables.
\end{enumerate}
These axioms are satisfied by the mutual information $I(A;B\vert C) = S(A,C)+S(B,C)-S(A,B,C)-S(C)$, where $S(A,C) = -\tr(\rho_{A,C}\log\rho_{A,C})$ denotes the von Neumann entropy of the joint system $A,C$ in the state $\rho_{A,C}$, and similarly for the other terms~\cite{NielsenChuang}. 

Based on the above axioms, we derive the following extension of the \ac{IC} principle~\cite{pawlowski2009information}, formulated as a lower bound on communication complexity.
\begin{restatable}{theorem}{ULB}\label{th:ULB}
    Let $f:X\times Y\to \Set{0,1}$ be a function computed distributively, and let $y^0,y^1,\dots,y^{\abs{Y}-1}$ be an ordering of the elements of $Y$.
    Then, the entanglement-assisted one-way classical communication complexity of $f$ satisfies
    \begin{equation}\label{eq:ULB}
        \CC_\eps(f)\geq \sum_{i=0}^{\abs{Y}-1}I(g;f(x,y^i)\vert \Set{f(x,y^j)}_{j=0}^{i-1},y=y^i),
    \end{equation}
    where the Shannon conditional mutual information is calculated for random variables $x$, $y$, and $g$ taking values in $X$, $Y$, and $\Set{0,1}$, respectively, satisfying $\Pr[g=f(i,j)\vert x=i,y=j]\geq 1-\eps$, for all $i\in X,j\in Y$, with $\eps\in [0,\frac12)$.
\end{restatable}

A proof of \cref{th:ULB} from the axioms (i-iv) of mutual information can be found in~\cref{app:ULB_proof}.
\Cref{th:ULB} holds for any ordering of the elements of $Y$, though, certain orderings can lead to simpler calculations or tighter lower bounds.
To demonstrate the power of the lower bound in \cref{eq:ULB}, we apply it to the functions defined in the preliminaries section.
\begin{example}\label{ex:INDEX}
$\CC_\eps(\INDEX_n)\in\Omega(n)$.
\end{example}
We start with the index function~\eqref{eq:def_index}, which corresponds to the original scenario of the \ac{IC} principle~\cite{pawlowski2009information}. 
Let us take the natural ordering of elements $0,1,\dots,n-1$ in $Y=[n]$, and let $x$ be uniformly distributed.
Since in this case $f(x,i) = x_i$, the values of $f(x,i)$ and $f(x,j)$ for $i\neq j$ are statistically independent, and the lower bound in \cref{th:ULB} simplifies to
\begin{equation}\label{eq:LB_index}
    \CC_\eps(\INDEX_n)\geq \sum_{i=0}^{n-1}I(g;x_i\vert y=i),
\end{equation}
where to get rid of the conditioning on $\Set{f(x,j)}_{j=0}^{i-1}$ for given $i$, one uses the chain rule and non-negativity of mutual information.
Using Fano's inequality~\cite{fano1968transmission} for each summand in \cref{eq:LB_index}, $I(g;x_i\vert y=i)\geq 1-h(\Pr[g=x_i\vert y=i])$, where $h(p)\coloneqq -p\log(p)-(1-p)\log(1-p)$ is the binary entropy, the assumption that $\Pr[g=x_i\vert y=i]=\frac{1}{2}\sum_{j=0}^1\Pr[g=j\vert x_i=j,y=i]\geq 1-\eps$, and the fact that the binary entropy is monotonically decreasing on the interval $(\frac{1}{2},1]$, we obtain that $\sum_{i=0}^{n-1}I(g;x_i\vert y=i)\geq n(1-h(\eps))$, and, therefore, $\CC_\eps(\INDEX_n)\in\Omega(n)$.
It is clear that the above analysis would not change if we took a different ordering of the elements of $Y$.
Similar analysis has been previously done for the index function in the context of streaming bounds~\cite{lovett2023} in Refs.~\cite{nayak1999,Ambainis1999}.

\begin{example}\label{ex:IP}
$\CC_\eps(\IP_n)\in\Omega(n)$.
\end{example}
Next, we consider the inner product function~\eqref{eq:def_ip}.
Let us choose an order of elements of $Y$, such that the first $n$ elements $y^0$,$y^1$, to $y^{n-1}$ are $100\dots 0,$ $010\dots0$, to $000\dots 1$.
For these inputs of Bob, $\IP_n(x,y^i)=x_i$ for $i\in[n]$.
As a result, the first $n$ terms in the sum in \cref{eq:ULB} are the same as for the index function in \cref{eq:LB_index}.
This is sufficient to conclude that $\CC_\eps(\IP_n)\in\Omega(n)$.
If one continues to compute the sum for the remaining $2^n-n$ elements of $Y$, all the conditional mutual information terms vanish, since conditioning on $\Set{x_i}_{i=0}^{n-1}$ is equivalent to conditioning on $x$ and subsequently the value of $f(x,y^i)$ in \cref{eq:ULB} is deterministic. 
It is worth noting that the two functions that we encountered,  $\IP_n$ and $\INDEX_n$ have different domains, yet they lead to the same form of the lower bound in \cref{eq:LB_index}. 

\begin{example}\label{ex:DISJ}
$\CC_\eps(\DISJ_n)\in\Omega(n)$.
\end{example}
Surprisingly, we can make a similar observation for the disjointness function. 
First, notice that one can express this function as $\DISJ_n(x,y) = \prod_{i=0}^{n-1}(1\oplus x_i\cdot y_i)$.
Again, we choose an order of the elements of $Y$ as we did for the inner product function. 
Then, for $y^i$, with $i\in[n]$, we have $\DISJ_n(x,y^i) = 1\oplus x_i$, and since the mutual information is symmetric with respect to the permutation of the values of random variables, we once again get the same expression for the lower bound as in \cref{eq:LB_index}.
Therefore, we conclude that $\CC_\eps(\DISJ_n)\in \Omega(n)$.
It is not clear whether an arbitrary ordering of $Y$ yields a lower bound that scales linearly with $n$, and even if it does, proving this is less straightforward.

\begin{example}\label{ex:EQ}
$\CC_0(\EQ_n)\geq n$.
\end{example}
Deriving deterministic communication complexity is, in general, simpler compared to the setting where errors are allowed.
This is also true for the lower bound in \cref{th:ULB}.
Indeed, setting $\eps=0$, means that $g=f(x,y^i)$ deterministically for all $i\in[\abs{Y}]$, and each of the conditional mutual information terms in \cref{eq:ULB} is equal to the entropy of $f(x,y^i)$.
To prove the bound for the equality function, we take a natural ordering $0,1,\dots,\abs{Y}-1$ of the elements of $Y$. 
For this ordering, knowing whether $f(x,j)$ is $0$ or $1$, for all $j<i$, is equivalent to knowing whether $x<i$ or $x\geq i$.
Therefore, the lower bound in \cref{eq:ULB} takes the form
\begin{equation}
    \CC_0(\EQ_n)\geq \sum_{i=0}^{2^n-1}\Pr[x\geq i]h(\Pr[x=i\vert x\geq i]),
\end{equation}
which, for the uniform distribution of $x$, is equal to $\sum_{i=0}^{2^n-1}(2^n-i)/2^nh(1/(2^n-i)) = \sum_{i=2}^{2^n}(i\log(i)-(i-1)\log(i-1))/2^n=n$.

As one would expect, for $\eps>0$, the derivations above do not lead to any non-trivial lower bound, which is consistent with the results of Ref.~\cite{kushilevitz1996communication} that $\CC_\eps(\EQ_n)\in \LandauO(1)$ for $\eps>0$ (see \cref{app:EQ_proof} for details).
At the same time, one can show a stronger result than the one above: for the one-sided error, i.e., if we allow for errors when $f(x,y)=0$, but still require perfect output for $f(x,y)=1$, or the other way around, the communication complexity scales linearly with $n$ (see \cref{app:EQ_proof} for details).

\begin{example}
    $\CC_\eps(\kINT_n)\in \Omega(n-2k)$.
\end{example}
Finally, we consider the $k$-intersect function introduced in \cref{eq:def_kint}. 
To prove the lower bound for $k\in\Set{2,3,\dots, \lfloor n/2\rfloor}$, we take an ordering of the elements of $Y=\Set{0,1}^n$ that starts with $\Set{y\in Y\given \abs{y}=k}$, where $\abs{y}$ is the Hamming weight of $y$.
A general idea is to reduce the expression in the lower bound~\eqref{eq:ULB} to multiple sums of the form in \cref{eq:LB_index}, which are multiplied by probabilities, which sum up to a constant in $n$ and $k$ (details can be found in \cref{app:kINT}).

\begin{table}[t!]
    \begin{tabular}{ |c|c| } 
        \hline
        Function $f$  & Result \\
        \hline
        $\INDEX_n$ $\IP_n$, $\DISJ_n$ & $\CC_\eps(f)\in \Theta(n)$ \\ 
        $\mathrm{EQ}_{n}$ & $\CC_0(f) = n$\\ 
        $\kINT_{n}$ & $\CC_\eps(f)\in \Omega(n-2k)$ \\ 
        \hline
    \end{tabular}
    \caption{Summary of the communication complexity bounds implied by \cref{th:ULB} for examples considered in this paper.}
    \label{tab:cc_ex}
\end{table}

We summarize the results of applying \cref{th:ULB} to the considered examples in \cref{tab:cc_ex}.
For the index, inner product, and disjointness functions, the lower bounds are asymptotically optimal, because for each, there exists a protocol in which Alice sends her input string using $\LandauO(n)$ bits. 
For the equality function, it is also clear that exactly $n$ bits of communication suffice for Bob to guess with unit probability whether his input equals Alice's input.
We also expect that the lower bound for the $k$-intersect function should scale linearly with $n$ and decrease with $k$.

Apart from choosing an ordering of $Y$, applying \cref{th:ULB} to any given function is straightforward and can be automated. All that is required is calculating the conditional probabilities of the type $\Pr[f(x,y^i)=k_i\vert \bigwedge_{j=0}^{i-1}f(x,y^j)=k_j]$ for $k_0,k_1,\dots,k_i\in \Set{0,1}$.
The joint distribution of $g$ and $f(x,y^i)$ is then determined by the condition of \cref{th:ULB}.
A computer program that computes the lower bound in \cref{eq:ULB} for a given function and given $n$ is available on GitHub~\cite{github_code}.

\subsection*{Bounding the set of quantum correlations}
So far, we have applied \cref{th:ULB} to draw predictions about the communication required to distributively compute functions in a physical theory satisfying the information axioms. 
However, we can also apply \cref{th:ULB} to bound the strength of nonlocal correlations in Bell experiments, similarly to the way the original \ac{IC} statement has been used~\cite{pawlowski2009information,allcock2009,miklin2021information,gachechiladze2022quantum,jain2024informationcausality,chaves2015information}.
Here, we argue that \cref{th:ULB} implies the principle of non-trivial communication complexity~\cite{brassard2006limit}, which has been previously shown to put constraints on the strength of correlations in Bell experiments~\cite{brassard2006limit,brunner2009nonlocality,botteron2024extending,botteron2024algebra,botteron2024communication}. 
Following Ref.~\cite{buhrmann2010nonlocality}, we say that communication complexity in a physical theory is trivial if for any function $f:X\times Y\to \Set{0,1}$, there exists $\eps\in [0,\frac12)$ independent of $n=\log(\abs{X})$, such that $\CC_\eps(f)\leq 1$ for all $n$.
The principle of non-trivial communication complexity is the negation of that and states that there must exist functions such that for any $\eps<\frac12$ we choose, there would be $n$ for which the communication complexity is greater than one bit.
Note that the above definition can be generalized to any constant-size communication, i.e., not necessarily one bit.
As \cref{ex:INDEX} shows, \cref{th:ULB} implies $\CC_\eps(\INDEX_n)\in \Omega(n)$ for any $\eps\in [1,\frac12)$, and, as a corollary, we can conclude the following.
\begin{observation}\label{res:ic_implies_ntcc}
	Information causality implies non-trivial communication complexity.
\end{observation}

Overall, one can understand \cref{th:ULB} as a unification of the two principles and also their extension.
In particular, \cref{th:ULB} fine-grains the principle of non-trivial communication complexity to inputs $x$ and $y$ of finite length.
As for the information causality principle, \cref{th:ULB} offers more choice in terms of the functions that the parties compute distributively. 
In fact, for the given cardinalities of the sets $X$ and $Y$, one can exhaustively check all Boolean functions $f:X\times Y\to \Set{0,1}$.
As we have seen in \cref{ex:INDEX,ex:IP,ex:DISJ}, some seemingly unrelated functions lead to the same expression in terms of the mutual information, thus \cref{th:ULB} induces a partition on the space of distributed functions.
Therefore, one can talk about equivalence classes of functions and compare their relative strength in bounding the set of quantum correlations.
In \cref{app:classification}, we provide such classifications for functions with $X=Y=\Set{0,1}^2$.

\section{Discussions and outlook}
In this work, we demonstrate that an extension of the information causality principle implies a lower bound on one-way communication complexity, applicable to both entanglement-assisted classical communication and quantum communication via the teleportation protocol.
We illustrate the effectiveness of this method through several examples, showing that it can yield tight lower bounds.
Importantly, the resulting lower bounds are theory-independent, since they rely solely on the axioms for mutual information.
From the point of view of bounding the set of quantum correlations in Bell experiments, we show that the proposed extended information causality principle is stronger than that of non-trivial communication complexity.

The present study opens up several avenues for future research at both the fundamental and technical levels.
First, it would be natural to try to extend the framework from one-way to two-way communication. It is known that for certain functions, such as the index and disjointness functions, the two cases differ significantly~\cite{buhrman1998quantum,aaronson2003}. 
If the two-way version of the universal lower bound can still be obtained in a theory-independent manner, it would provide further generalizations of the information causality principle.

We have observed in \cref{ex:INDEX,ex:IP,ex:DISJ}, that different functions may allow for the same expression of the bound in \cref{eq:ULB}.
At the same time, the number of \ac{PR}-boxes required to distributively compute their value using one bit of communication, following the construction of van Dam~\cite{vandam2013implausible}, can be different (see \cref{app:convolution} for details).
A preliminary analysis suggests that this has implications for the relative strength of constraints one can derive on the set of quantum correlations from \cref{th:ULB}, with functions requiring less number of \ac{PR}-boxes leading to tighter bounds (see \cref{app:convolution} for details).
However, we believe that this research question deserves a separate study.

On a more technical level, several open questions remain. First, how can the bound given by \cref{th:ULB} be calculated or approximated more efficiently for a given function and input size $n$? 
Second, how can one determine the optimal ordering of Bob's inputs in \cref{eq:ULB} for a given function? 
It would also be interesting to explore applications of the lower bound to other functions, as well as to extend the classification of the lower bounds for Boolean functions to input sizes larger than two.

\acknowledgements
    We thank Miguel Navascu\'es, Lucas Vieira, Simon Pfeiffer, and Anna Schroeder for useful discussions.
    This research was funded by the Fujitsu Germany GmbH as part of the endowed professorship ``Quantum Inspired and Quantum Optimization''.
    This project was funded within the QuantERA II Programme that has received funding from the EU's H2020 research and innovation programme under the GA No 101017733.
\newpage 
\onecolumngrid

\section*{Appendix}
\appendix
\crefalias{section}{appendix}
In this Appendix, we give technical details that support the claims in the main text. 
In~\cref{app:ULB_proof}, we provide a proof of \cref{th:ULB}.
In~\cref{app:EQ_proof}, and \cref{app:kINT} we give details on the application of \cref{th:ULB} to the equality and $k$-intersect functions, respectively.
In~\cref{app:classification}, we provide classification of all Boolean functions $f:X\times Y\to\Set{0,1}$ with equivalent extended \ac{IC} statements for $X=Y=\Set{0,1}^2$.
In~\cref{app:convolution}, we discuss how different functions from the same class may impose constraints of different relative strength on the set of quantum correlations.

\section{Proof of the extended \ac{IC} statement}\label{app:ULB_proof}
We start by restating the result.
\ULB*
\begin{proof}
The proof proceeds along the same lines as the proof of the Information Causality principle~\cite{pawlowski2009information} (see also Ref.~\cite{jain2024informationcausality}).
Our starting point is the following set of axioms:
\begin{enumerate}
	\item[(i)] non-negativity: $I(A;B\vert C)\geq 0$,
	\item[(ii)] chain rule: $I(A;B\vert C) = I(A;B,C)-I(A;C)$,
	\item[(iii)] data processing: $I(A;B\vert C) \geq I(\Phi(A);B,C)$, for a physical map $\Phi$,
	\item[(iv)] and consistency with Shannon mutual information in case $A,B$ and $C$ are (classical) random variables,
\end{enumerate}
which we assume to hold for a measure of conditional mutual information $I(A;B\vert C)$ for subsystems $A,B$, and $C$, which we can think of as parts of a quantum system being in some state $\rho_{A,B,C}$.

Let $x$ and $y$ be random variables taking values in $X$ and $Y$, representing Alice's and Bob's inputs, respectively (see \cref{fig:scenario}).
Let the parties share a source of correlations (e.g., a quantum state $\rho_{A,B}$) with its parts denoted as $A$ and $B$.
In the entanglement-assisted one-way communication scenario, Alice sends Bob a classical message, which we denote as $M$, which is a random variable that depends on $A$ and $x$.
Bob's guess $g$ is a random variable that depends on $y$, $M$, and $B$.

The derivation consists of re-writing the finding an upper and a lower bound on a quantity $I(\Set{f(x,y^i)}_{i=0}^{\abs{Y}-1};M,B)$, for an ordering $y^0,y^1,\dots,y^{\abs{Y}}-1$ of the elements of $Y$. 
We start by rewriting this quantity using the chain rule,
    \begin{align}\label{eq:ULB_der_1}
        I(\Set{f(x,y^i)}_{i=0}^{\abs{Y}-1};M,B)&=I(f(x,y^0);M,B)+I(\Set{f(x,y^i)}_{i=1}^{\abs{Y}-1};M,B\vert f(x,y^0))\\
        &=I(f(x,y^0);M,B)+I(f(x,y^1);M,B\vert f(x,y^0))+I(\Set{f(x,y^i)}_{i=2}^{\abs{Y}-1};M,B\vert \Set{f(x,y^i)}_{i=0}^1)\nonumber\\
        &=\dots=\sum_{i=0}^{\abs{Y}-1}I(f(x,y^i);M,B\vert \Set{f(x,y^j)}_{j=0}^{i-1})=\sum_{i=0}^{\abs{Y}-1}I(f(x,y^i);M,B\vert \Set{f(x,y^j)}_{j=0}^{i-1},y=y^i).\nonumber
    \end{align}

In the last step in \cref{eq:ULB_der_1}, we have added conditioning on a particular value of $y$ in each term. 
We could do that because all the variables in the mutual information are independent of the random variable $y$: $B$ is independent of $y$, because it is a part of the shared system and not the result of a measurement of Bob, 
$f(x,y^i)$ is a function of $x$ only, and $M$ is independent of $y$ due to the non-signalling condition. 
Now, we apply the data processing inequality for each summand of the final expression in \cref{eq:ULB_der_1} for a map $\Phi_{y_i}(M,B)=g$, which can be different for each $y^i$.  
As a result, we obtain that the quantity $I(\Set{f(x,y^i)}_{i=0}^{\abs{Y}-1};M,B)$ is lower-bounded by the expression on the right-hand side of \cref{eq:ULB}.

To find an upper bound, we use the data processing inequality for a map $\Phi(x)=\Set{f(x,y^i)}_{i=0}^{\abs{Y}-1}$, obtaining
\begin{equation}\label{eq:ULB_der_2}
    I(\Set{f(x,y^i)}_{i=0}^{\abs{Y}-1};M,B)\leq I(x;M,B).
\end{equation}
Now, we apply the chain rule (twice) and the non-negativity of the mutual information to bound $I(x;M,B)$ from above as
\begin{equation}\label{eq:eq:ULB_der_3}
    I(x;M,B) = I(M;x,B)+I(x;B)-I(M;B)\leq I(M;x,B),    
\end{equation}
where $I(x;B)=0$ due to the non-signalling of the correlations.
Finally, we realize that the mutual information of the message $M$ with any other variable cannot exceed its size, which we denoted in the main text as $m$, i.e., we have established that
\begin{equation}\label{eq:ULB_der_3}
    m\geq \sum_{i=0}^{\abs{Y}-1}I(g;f(x,y^i)\vert \Set{f(x,y^j)}_{j=0}^{i-1},y=y^i).
\end{equation}
If we assume that $\Pr[g=f(i,j)\vert x=i,y=j]\geq 1-\eps$ for all $i\in X$ and $j\in Y$, as in the statement of the theorem, then the minimum value of $m$ in the definition of the communication complexity in \cref{eq:def_cc} must also satisfy~\cref{eq:ULB_der_3}, which completes the proof.
\end{proof}
A similar expression as in \cref{eq:ULB_der_3} has appeared previously in the context of quantum communication complexity in the Appendix of Ref.~\cite{zhang2011}, albeit without a proof.

\section{Details on the application of \cref{th:ULB} to the equality function}\label{app:EQ_proof}
First, let us first define a more general notion of the communication complexity for a function $f:X\times Y\to \Set{0,1}$ that differentiates between errors for the cases when the function's value is $0$ or $1$,
\begin{align}\label{eq:def_cc_asym}
	\CC_{(\epsi,\epsii)}(f)\coloneqq \min~& m\\
    \text{s.t.}~ & \Pr[g=0\vert x=i,y=j]\geq 1-\epsi,\; \text{for all}~ i,j,~ \text{s.t.~} f(i,j)=0,\nonumber\\
    & \Pr[g=1\vert x=i,y=j]\geq 1-\epsii,\; \text{for all}~ i,j,~ \text{s.t.~} f(i,j)=1,\nonumber
\end{align}
where the parameters $\epsi,\epsii\in [0,\frac12)$ are sometimes referred to as Type-I and Type-II errors.
Clearly, the lower bound in \cref{eq:ULB} also applies to $\CC_{(\epsi,\epsii)}$ if we substitute the condition on the guessing probability in the statement of \cref{th:ULB} with the one in \cref{eq:def_cc_asym}. 
It must also be clear that $\CC_0(f)\geq \CC_{(\epsi,0)}(f)$ and $\CC_0(f)\geq \CC_{(0,\epsii)}(f)$ for any function $f$.

For the equality function, we derive a stronger result than the one given in the main text. 
\begin{proposition} $\CC_{(0,\epsii)}(\EQ_n)\in \Omega(n)$, and $\CC_{(\epsi,0)}(\EQ_n)\in \Omega(n)$.
\end{proposition}
\begin{proof}
Let us take the natural ordering $0,1,\dots,2^n-1$ of the elements of $Y=[2^n]$ in \cref{th:ULB}.
We keep the first term in the sum on the right-hand side of \cref{eq:ULB} simply as is, $I(g;f(x,0)\vert y=0)$.
All other terms for $i\in \Set{1,2,\dots 2^n-1}$ we expand as
\begin{equation}\label{eq:LB_EQ_dev_1}
    I(g;f(x,i)\vert \Set{f(x,j)}_{j=0}^{i-1},y=i) = \sum_{z\in\Set{0,1}^i}\Pr[\bigwedge_{j=0}^{i-1} f(x,j)=z_j]I(g;f(x,i)\vert \bigwedge_{j=0}^{i-1} f(x,j)=z_j,y=i),
\end{equation}
where $z_j$ for $j\in \Set{0,1,\dots i-1}$ are the elements of $z$. 
Evidently, whenever $\sum_{j=0}^{i-1}z_j>1$, the probability in~\cref{eq:LB_EQ_dev_1} is $0$.
Whenever $\sum_{j=0}^{i-1}z_j=1$, we know that $f(x,i)=0$, and, hence the mutual information in~\cref{eq:LB_EQ_dev_1} is $0$. 
Therefore, the only term that survives in the above sum is the one for which $z_j=0$ for all $j\in [i]$.
Notice, that we can conveniently represent the condition $\bigwedge_{j=0}^{i-1} f(x,j)=0$ as $x\geq i$, which leads to the following lower bound
\begin{equation}\label{eq:LB_EQ_dev_2}
    \CC_{(\epsi,\epsii)}(\EQ_n)\geq \sum_{i=0}^{2^n-2}\Pr[x\geq i]I(g;[x=i]\vert x\geq i,y=i),
\end{equation}
where, of course, $\Pr[x\geq 0]I(g;[x=0]\vert x\geq 0, y=0)=I(g;[x=0]\vert y=0)$.
Notice that we set the limit of the sum in \cref{eq:LB_EQ_dev_2} to $2^n-2$, because for $i=2^n-1$, the condition $x\geq i$ is equivalent to $x=2^n-1$, and the function's value is deterministic.
Taking the uniform distribution of $x\in[2^n]$, we get $\Pr[x\geq i] = \frac{2^n-i}{2^n}$.
We express the mutual information terms for $i\in\Set{0,1,\dots,2^n-2}$ in \cref{eq:LB_EQ_dev_2} in terms of the binary entropy as
\begin{multline}\label{eq:LB_EQ_dev_3}
    I(g;[x=i]\vert x\geq i, y=i) = h(\Pr[g=0\vert x\geq i,y=i])-\Pr[x=i\vert x\geq i]h(\Pr[g=0\vert x=i,y=i])\\
    -\Pr[x\geq i+1\vert x\geq i]h(\Pr[g=0\vert x\geq i+1,y=i]).
\end{multline}
The probability $\Pr[g=0\vert x\geq i,y=i]$ for $i\in\Set{0,1,\dots,2^n-2}$ can be expressed as
\begin{equation}\label{eq:LB_EQ_dev_4}
\begin{split}
    \Pr[g=0\vert x\geq i,y=i] & = \Pr[g=0,x=i\vert x\geq i,y=i]+\Pr[g=0,x\geq i+1\vert x\geq i,y=i] \\
    & = \Pr[x=i\vert x\geq i]\Pr[g=0\vert x=i,y=i]+\Pr[x\geq i+1\vert x\geq i]\Pr[g=0\vert x\geq i+1,y=i].
\end{split}
\end{equation}
Therefore, the lower bound in \cref{eq:ULB} on the communication complexity of $\EQ_n$, for the chosen ordering of the elements of $Y$ is completely determined by the two types of probabilities $\Pr[g=0\vert x=i,y=i]$ and $\Pr[g=0\vert x\geq i+1,y=i]$, together with the probabilities involving only variable $x$, i.e., $\Pr[x=i\vert x\geq i]=\frac{1}{2^n-i}$, and $\Pr[x\geq i+1\vert x\geq i]=1-\frac{1}{2^n-i}$.

Now we consider the case of $\epsi=0$.
From the condition in \cref{eq:def_cc_asym}, we get that $\Pr[g=0\vert x\geq i+1,y=i]=1$, for all $i\in [2^n-1]$.
Thus, the mutual information in \cref{eq:LB_EQ_dev_3} takes the form
\begin{equation}
    I(g;[x=i]\vert x\geq i, y=i) = h\left(\frac{1}{2^n-i}\Pr[g=1\vert x=i,y=i]\right)-\frac{1}{2^n-i}h\left(\Pr[g=1\vert x=i,y=i]\right).
\end{equation}
It is relatively easy to see that the above expression as a function of $\Pr[g=1\vert x=i,y=i]$ is monotonically increasing.
Since from \cref{eq:def_cc_asym} we know that $\Pr[g=1\vert x=i,y=i]\geq 1-\epsii$ for all $i\in[2^n-1]$, we can conclude that 
\begin{equation}
    I(g;[x=i]\vert x\geq i, y=i)\geq h\left(\frac{1-\epsii}{2^n-i}\right)-\frac{h\left(\epsii\right)}{2^n-i}.
\end{equation}
Inserting this back into the lower bound in \cref{eq:LB_EQ_dev_2}, we obtain
\begin{equation}\label{eq:LB_EQ_dev_5}
    \CC_{(0,\epsii)}(\EQ_n)\geq \sum_{i=0}^{2^n-2}\frac{2^n-i}{2^n}h\left(\frac{1-\epsii}{2^n-i}\right)-\frac{2^n-1}{2^n}h\left(\epsii\right).
\end{equation}
As we are interested in the asymptotic behaviour, we can ignore the constant terms, even if they are negative.
From the above sum, we can take only the first summand of the binary entropy function and obtain
\begin{equation}
    \CC_{(0,\epsii)}(\EQ_n)\geq \frac{1-\epsii}{2^n}\sum_{i=0}^{2^n-2}\log(2^n-i)-\LandauO(1)=\frac{1-\epsii}{2^n}\log(2^n!)-\LandauO(1)\geq (1-\epsii)\log(2^n)-\LandauO(1)\in \Omega(n).
\end{equation}
The other case of $\epsii=0$ follows the same reasoning and also leads to the linear lower bound.
\end{proof}

Next, we show that if we continued the derivation above for $\epsi>0$, we would not obtain any non-trivial lower bound on the communication complexity, which is consistent with the results of Ref.~\cite{kushilevitz1996communication}. To show this, we consider the symmetric case of $\epsi=\epsii=\eps$ for simplicity.
We also take the equality in the definition of communication complexity in \cref{eq:def_cc}, i.e., we assume $\Pr[g=f(i,j)\vert x=i,y=j]=1-\eps$, $\forall i\in X,j\in Y$, because the mutual information in \cref{eq:LB_EQ_dev_3} is maximal in that case.

\begin{proposition}
    Let $x$ and $y$ take values in $[2^n]$ and $x$ be uniformly distributed. Let $\Pr[g=[i=j]\vert x=i,y=j]=1-\eps$, for all $i,j\in [2^n]$ and $\eps\in (0,\frac12)$.
    Then
    \begin{equation}\label{eq:LB_EQ_dev_7}
        \sum_{i=0}^{2^n-2}\Pr[x\geq i]I(g;[x=i]\vert x\geq i,y=i)\in \LandauO(1).
    \end{equation}
\end{proposition}
\begin{proof}
Let us denote the expression on the left-hand side of \cref{eq:LB_EQ_dev_7} as $I$.
Inserting $\Pr[g=[i=j]\vert x=i,y=j]=1-\eps$ in \cref{eq:LB_EQ_dev_4,eq:LB_EQ_dev_3}, we obtain
\begin{equation}
    \Pr[g=0\vert x\geq i,y=i] = \frac{\eps}{2^n-i}+\left(1-\frac{1}{2^n-i}\right)(1-\eps)=1-\eps-\frac{1-2\eps}{2^n-i},
\end{equation}
and
\begin{equation}
    I(g;[x=i]\vert x\geq i,y=i) = h\left(\eps+\frac{1-2\eps}{2^n-i}\right)-h\left(\eps\right).
\end{equation}
Thus, the resulting sum in \cref{eq:LB_EQ_dev_7} becomes
\begin{equation}\label{eq:LB_EQ_dev_8}
    I=\frac{1}{2^n}\sum_{i=2}^{2^n}i\left(h\left(\eps+\frac{1-2\eps}{i}\right)-h\left(\eps\right)\right),
\end{equation}
where we changed the index of summation from $i$ to $2^n-i$.
For each summand in \cref{eq:LB_EQ_dev_8}, we can bound the first binary entropy term as
\begin{equation}
    h\left(\eps+\frac{1-2\eps}{i}\right) \leq h\left(\eps\right)+\frac{1-2\eps}{i}\log\left(\frac{1-\eps}{\eps}\right),
\end{equation}
which follows from the non-negativity of the Kullback–Leibler divergence.
We can, therefore, conclude that 
$ I\leq (1-2\eps)\log\left(\frac{1-\eps}{\eps}\right),$
which completes the proof.
\end{proof}

\section{Details on the application of \cref{th:ULB} to the $k$-intersect function}
\label{app:kINT}
\begin{proposition}\label{prop:kINT}
    $\CC_\eps(\kINT_n)\in \Omega(n-2k)$, for $k\in\Set{1,2,\dots,\lfloor n/2\rfloor}$.
\end{proposition}
\begin{proof}
To simplify the notation, let us denote by $x_{i:j} \coloneqq x_ix_{i+1}\dots x_j$ for $0\leq i\leq j\leq n-1$ the substring of $x\in \Set{0,1}^n$ of length $j-i+1$, and similarly for $y$.
We use the ordering of the elements of $Y$, such that the first $n\choose k$ of them have the Hamming weight $\abs{y}=k$.
In fact, we only look at such terms with $\abs{y}=k$ in the sum in \cref{eq:ULB}, as we can always ignore the rest.

Within the $n\choose k$ elements of $Y$, first we consider $n-k+1$ of them $y^0,y^1,\dots, y^{n-k}$, with $\abs{y^i_{0:k-2}}=k-1$, and $y^i_{k+i-1}=1$, i.e., the strings that start with $k-1$ bits being $1$ and have another bit equal to $1$ at the $(k+i-1)$-th position, for $i\in \Set{0,1,\dots,n-k}$. For these $n-k+1$ elements of $Y$, the function $f(x,y)=\kINT_n(x,y)$ takes the form $f(x,y^i)=[\sum_{j=0}^{n-1}x_j\cdot y^i_j\geq k] = [\abs{x_{0:k-2}}+x_{k+i-1}=k]$, and we obtain the following expression for the corresponding part of the lower bound in \cref{eq:ULB},
\begin{equation}\begin{split}\label{eq:kINT_der_1}
    \sum_{i=0}^{n-k}I(g;f(x,y^i)\vert \Set{f(x,y^j)}_{j=0}^{i-1},y=y^i) &\geq \sum_{i=1}^{n-k}I(g;f(x,y^i)\vert [\abs{x_{0:k-1}}=k], \Set{f(x,y^j)}_{j=1}^{i-1},y=y^i)\\
    &= \Pr[\abs{x_{0:k-1}}=k]\sum_{i=1}^{n-k}I(g;f(x,y^i)\vert \abs{x_{0:k-1}}=k, \Set{f(x,y^j)}_{j=1}^{i-1},y=y^i) \\
    &+ \Pr[\abs{x_{0:k-1}}<k]\sum_{i=1}^{n-k}I(g;f(x,y^i)\vert \abs{x_{0:k-1}}<k, \Set{f(x,y^j)}_{j=1}^{i-1},y=y^i),
    \end{split}
\end{equation}
where we ignored the first term $I(g;[\abs{x_{0:k-1}}=k]\vert y=y^0)$.
Notice that the condition $\abs{x_{0:k-1}}=k$ implies $\abs{x_{0:k-2}}=k-1$, and hence, for the mutual information conditioned on $\abs{x_{0:k-1}}=k$ the function $f(x,y^i) = x_{k+i-1}$, for $i\in \Set{1,2,\dots, n-k}$. 
Therefore, the sum multiplied by the probability $\Pr[\abs{x_{0:k-1}}=k]=\frac{1}{2^k}$ in \cref{eq:kINT_der_1} takes the form
\begin{equation}\label{eq:kINT_der_2}
    \sum_{i=1}^{n-k}I(g;x_{k+i-1}\vert \abs{x_{0:k-1}}=k,\Set{x_{k+j-1}}_{j=1}^{i-1},y=y^i)\geq \sum_{i=1}^{n-k}I(g;x_{k+i-1}\vert y=y^i)\geq (n-k)\left(1-h\left(\eps\right)\right),
\end{equation}
where the conditioning can be removed due to the independence of the bits of $x$.
Now we deal with the sum in \cref{eq:kINT_der_1} that is multiplied by $\Pr[\abs{x_{0:k-1}}<k]$, which we call the second summand. 
We open it up similarly to \cref{eq:kINT_der_1}, but now with respect to the value of the function $f(x,y^1)$.
The resulting expression is
\begin{equation}\begin{split}\label{eq:kINT_der_3}
    &\Pr[\abs{x_{0:k-2}}+x_k=k\vert \abs{x_{0:k-1}}<k]\sum_{i=2}^{n-k}I(g;f(x,y^i)\vert \abs{x_{0:k-1}}<k, \abs{x_{0:k-2}}+x_k=k, \Set{f(x,y^j)}_{j=2}^{i-1},y=y^i)\\
    +&\Pr[\abs{x_{0:k-2}}+x_k<k\vert \abs{x_{0:k-1}}<k]\sum_{i=2}^{n-k}I(g;f(x,y^i)\vert \abs{x_{0:k-1}}<k, \abs{x_{0:k-2}}+x_k<k, \Set{f(x,y^j)}_{j=2}^{i-1},y=y^i).
    \end{split}
\end{equation}
Perhaps the reader can already anticipate that the first term in \cref{eq:kINT_der_3} simplifies in a similar manner to \cref{eq:kINT_der_2}. We call this step of separating the second summand a \textit{splitting} into a lower-bound and the remaining exact terms.  
In \cref{eq:kINT_der_3} , due to the condition $\abs{x_{0:k-2}}+x_k=k$ in the mutual information, we again obtain that $f(x,y^i)=x_{k+i-1}$ for $i\in\Set{2,3,\dots,n-k}$.
The only difference with \cref{eq:kINT_der_2} is that the sum starts at $i=2$ and, hence, the respective lower bound is $(n-k-1)(1-h(\eps))$.
The factor by which this sum is multiplied is 
\begin{equation}
    \Pr[\abs{x_{0:k-2}}+x_k=k\vert \abs{x_{0:k-1}}<k]\Pr[\abs{x_{0:k-1}}<k] = \Pr[\abs{x_{0:k-2}}+x_k=k\wedge \abs{x_{0:k-1}}<k]=\frac{1}{2^{k+1}}.
\end{equation}
Clearly, the above procedure can be repeated for $y^2,y^3,\dots,y^{n-k-1}$, with the last sum consisting of just one term corresponding to $y=y^{n-k}$ and lower-bounded by $1-h(\eps)$, multiplied by the probability $\Pr[\abs{x_{0:k-2}}+x_{n-2}=k \bigwedge_{i=k-1}^{n-3}\abs{x_{0:k-2}} + x_i <k] = \frac{1}{2^{n-1}}$.
Now, notice that if the value of the function $f(x,y^{n-k})=1$, i.e., if $\abs{x_{0:k-2}}+x_{n-1}=k$, then together with the conditions $\abs{x_{0:k-2}}+x_i$ for $i\in\Set{k-1,k,\dots,n-2}$ they single-out exactly one string $x$, the one with bits from $0$ to $k-2$, and $x_{n-1}$ being equal to $1$ and the rest being $0$.
This means that any mutual information term in \cref{eq:ULB}, conditioned on this configuration, is going to be $0$, because the value of the function is deterministic in $x$ and $y$.
Therefore, we can intermediately conclude the following lower bound
\begin{equation}\begin{split}\label{eq:kINT_der_4}
    \CC_\eps(\kINT_n)&\geq \left(1-h\left(\eps\right)\right)\sum_{i=0}^{n-k-1}\frac{1}{2^{k+i}}(n-k-i)\\
    +&\Pr[\abs{x_{0:k-2}}<k-1]\sum_{i=n-k+1}^{2^n-1}I(g;f(x,y^i)\vert \abs{x_{0:k-2}}<k-1, \Set{f(x,y^j)}_{j=n-k+1}^{i-1},y=y^i),
\end{split}
\end{equation}
where we have substituted the condition $\bigwedge_{i=k-1}^{n-1} \abs{x_{0:k-2}}+x_i<k$ with $\abs{x_{0:k-2}}<k-1$, which are equivalent up to one string $x$ with $\abs{x}=\abs{x_{0:k-2}}=k-1$, that does not play any role in the future calculations.

\begin{table}[t!]
    \centering
    \begin{tabular}{c c}
        $y^0=$ & $\overbrace{11\dots 1}^{k-1}$ $\overbrace{10\dots 0}^{n-k+1}$  \\
        $y^1=$ & $11\dots 1$ $01\dots 0$ \\
        $\vdots$ & $\vdots$ \\
        $y^{n-k}=$ & $11\dots 1$ $00\dots 1$
    \end{tabular} \qquad
    \begin{tabular}{c c}
        $y^{n-k+1}=$ & $\overbrace{11\dots 1}^{k-2}$ $01$ $\overbrace{10\dots 0}^{n-k}$ \\
        $y^{n-k+2}=$ & $11\dots 1$ $01$ $01\dots 0$ \\
        $\vdots$ & $\vdots$ \\
        $y^{2n-2k}=$ & $11\dots 1$ $01$ $00\dots 1$
    \end{tabular} \qquad
    \begin{tabular}{c c}
        $y^{2n-2k+1}=$ & $\overbrace{11\dots 1}^{k-2}$ $001$ $\overbrace{10\dots 0}^{n-k-1}$ \\
        $y^{2n-2k+2}=$ & $11\dots 1$ $001$ $01\dots 0$ \\
        $\vdots$ & $\vdots$ \\
        $y^{3n-3k-1}=$ & $11\dots 1$ $001$ $00\dots 1$
    \end{tabular}
    \caption{Sequences of elements of $Y$ considered in the proof of the lower bound on $\CC_\eps(\kINT_n)$.}
    \label{tab:kINT_orderings}
\end{table}

The next sequence of elements in $Y$ we consider are $y^{i}$ for $i\in \Set{n-k+1,n-k+2,\dots 2n-2k}$ with $\abs{y^{i}_{0:k-3}}=k-2$, $y^i_{k-1}=1$ and $y^i_{k+(i-n+k-1)}=1$.
The difference with the previous sequence is that we have a conditioning on $\abs{x_{0:k-2}}<k-1$, which, together with the condition $f(x,y^i)=1$ for $i\in \Set{n-k+1,n-k+2,\dots 2n-2k}$, imply $\abs{x_{0:k-3}}=k-2$, $x_{k-1}=1$, and $x_{k-2}=0$. 
As a result, we obtain the same expression as the first sum in \cref{eq:kINT_der_4}, with the only difference that it begins with $i=1$.
Notice here, that if we kept the condition $\bigwedge_{i=k-1}^{n-1} \abs{x_{0:k-2}}+x_i<k$ instead of the simplified one $\abs{x_{0:k-2}}<k-1$, the above argument would not change.

Continuing with the construction, the next sequences of elements in $Y$ we consider are $y^{i}$ for $i\in \Set{2n-2k+1,2n-2k+2,\dots 3n-3k-1}$ with $\abs{y^{i}_{0:k-3}}=k-2$, $y^i_{k}=1$ and $y^i_{k+(i-2n+2k)}=1$, then $y^{i}$ for $i\in \Set{3n-3k,3n-3k+1,\dots 4n-4k-3}$ with $\abs{y^{i}_{0:k-3}}=k-2$, $y^i_{k+1}=1$ and $y^i_{k+(i-3n+3k+2)}=1$, and so on.
We provide \cref{tab:kINT_orderings} to help the reader navigate the sequences that we consider in this proof.
It should be clear that with each new type of sequence, we get a lower bound of the type of the first sum in \cref{eq:kINT_der_4}, but each time with the index starting from one element higher. 
Therefore, we can conclude that 
\begin{equation}\begin{split}\label{eq:kINT_der_5}
    \CC_\eps(\kINT_n)&\geq \left(1-h\left(\eps\right)\right)\sum_{j=0}^{n-k-1}\sum_{i=j}^{n-k-1}\frac{1}{2^{k+i}}(n-k-i)\\
    +&\Pr[\abs{x_{0:k-3}}<k-2]\sum_{i=\mathrm{ind}}^{2^n-1}I(g;f(x,y^i)\vert \abs{x_{0:k-3}}<k-2, \Set{f(x,y^j)}_{j=\mathrm{ind}}^{i-1},y=y^i),
\end{split}
\end{equation}
where we have again substituted the condition $\bigwedge_{i=k-2}^{n-1} \abs{x_{0:k-3}}+x_i<k-1$ with $\abs{x_{0:k-3}}<k-2$ for the same reason as above, and the index $\mathrm{ind}=\frac{(n-k+2)(n-k+1)}{2}$ is the next index for $y^i$, because for each sequence corresponding to a column in \cref{tab:kINT_orderings}, which we can enumerate with $j=0,1,\dots$, we have the first $k-1+j$ bits of $y$ fixed, and $n-k+1-j$ bits used to index a particular bit of $x$.
The number of such sequences is $n-k+1$, with the last sequence for $j=n-k$ having only one bit string.

Now the pattern shall become obvious. In~\cref{eq:kINT_der_5}, we have two main summands, same as in~\cref{eq:kINT_der_4}: The second one is the exact expression, but the first one lower bounds the terms 

Comparing \cref{eq:kINT_der_1}, \cref{eq:kINT_der_4} and \cref{eq:kINT_der_5}, one can notice a clear pattern.
At every iteration, the string $x_{0:k-i}$ in the condition $\abs{x_{0:k-i}}<k+1-i$, for $i=1,2,\dots$, gets shorter by one bit.
As a result, the lower-bounded part, which is the first summand in \cref{eq:kINT_der_1,eq:kINT_der_4,eq:kINT_der_5}, gets an additional summation.
Once we reach the conditioning $\abs{x_{0:k-i}}\leq k+1-i$ for $i=k$, we clearly do not have any terms left to bound.
At that point, we reach the following lower bound
\begin{equation}\label{eq:kINT_der_6}
    \CC_\eps(\kINT_n)\geq \left(1-h\left(\eps\right)\right)\sum_{j_{k-1}=0}^{n-k-1}\sum_{j_{k-2}=j_{k-1}}^{n-k-1}\cdots \sum_{j_1=j_2}^{n-k-1}\sum_{i=j_1}^{n-k-1}\frac{1}{2^{k+i}}(n-k-i),
\end{equation}
where the number of sums is $k$.
Now we show that the nested sum above is lower-bounded by $n-2k$.
First, we prove by induction that 
\begin{equation}
    \sum_{j_{k-1}=j}^m\sum_{j_{k-2}=j_{k-1}}^m\dots \sum_{j_1=j_2}^{m}\sum_{i=j_1}^m c_i = \sum_{i=j}^m c_i {i-j+k-1 \choose k-1},
\end{equation}
for some $\Set{c_i}_{i=0}^{m}$, where again the number of the nested sums on the left-hand side is $k$.
The base case of $k=1$ holds trivially.
Assuming for $k$, we prove it for $k+1$:
\begin{equation}
    \sum_{j=l}^m\sum_{i=j}^m c_i {i-j+k-1 \choose k-1}=\sum_{i=l}^mc_i\sum_{j=0}^{i-l}{k-1+j \choose j}=\sum_{i=l}^mc_i{i-l+k\choose k}.
\end{equation}
Therefore, inserting the values $l=0$, $m=n-k-1$, and the coefficients $c_i$ from the lower bound~\cref{eq:kINT_der_6}, we obtain
\begin{equation}
    \sum_{i=0}^{n-k-1}\frac{1}{2^{k+i}}{i+k-1\choose k-1}(n-k-i)=\sum_{i=k}^{n-1}\frac{1}{2^{i}}{i-1\choose k-1}(n-i)\geq n-2k,
\end{equation}
where the last inequality can be proven by induction in $k$, as we show below.
For the base case of $k=1$, we have
\begin{equation}
    \sum_{i=1}^{n-1}\frac{1}{2^i}(n-i) = n-2+\frac{2}{2^n}\geq n-2.
\end{equation}
The induction step from $k$ to $k+1$ is
\begin{align}
    \sum_{i=k+1}^{n-1}\frac{1}{2^{i}}{i-1\choose k}(n-i)&=\sum_{i=k}^{n-2}\frac{1}{2^{i+1}}{i\choose k}(n-i-1)=\frac{1}{2}\sum_{i=k}^{n-2}\frac{1}{2^{i}}{i-1\choose k-1}(n-1-i)+\frac{1}{2}\sum_{i=k}^{n-2}\frac{1}{2^{i}}{i-1\choose k}(n-1-i)\nonumber\\
    &\geq \frac{1}{2}(n-1-2k)+\frac{1}{2}\sum_{i=k+1}^{n-2}\frac{1}{2^{i}}{i-1\choose k}(n-1-i).\label{eq:kINT_der_7}
\end{align}
Notice that the second summand in \cref{eq:kINT_der_7} is the same as the starting expression, but for $n-1$. 
Therefore, 
\begin{equation}
    \sum_{i=k+1}^{n-1}\frac{1}{2^{i}}{i-1\choose k}(n-i)\geq \sum_{j=1}^{n-2k}(n-2k-j)\frac{1}{2^j}=n-2(k+1)+\frac{2}{2^{n-2k}}\geq n-2(k+1),
\end{equation}
which completes the proof.
\end{proof}

\section{Classification of Boolean functions of input length 2 under extended IC}\label{app:classification}

The extended \ac{IC} statement in \cref{th:ULB} induces a partition of the set of distributed functions into equivalence classes, where two functions belong in the same class if they result in the same expression in \cref{eq:ULB} up to some relabelling of random variables and the sets in which they take their values. 
Here, we provide a complete classification of the set of all distributed functions for $X=Y=\Set{0,1}^2$.

Let Alice' and Bob's inputs be $x=x_0x_1$ and $y=y_0y_1$, respectively. We begin by writing a parametrised form for a general function by breaking apart the function into subfunctions which are conditioned on the values of $y$.
Let us denote $f_{y_0y_1}(x_0,x_1)\coloneqq f(x_0x_1,y_0y_1)$, then any function $f:\Set{0,1}^2\times\Set{0,1}^2\to\Set{0,1}$ can be expressed in the following form
\begin{equation}\label{eq:f2_gen}
    f(x_0x_1,y_0y_1)=(y_0\oplus1)(y_1\oplus1)f_{00}(x_0,x_1)\oplus(y_0\oplus1)y_1f_{01}(x_0,x_1)\oplus y_0(y_1\oplus1)f_{10}(x_0,x_1)\oplus y_0\cdot y_1f_{11}(x_0,x_1).
\end{equation}
Each of the functions $f_{y_0y_1}$, in turn, can be restricted to be from the following set
\begin{equation}\label{eq:fy0y1set}
   f_{y_0y_1}(x_0,x_1)\in \Set{0,x_0,x_1,x_0\cdot x_1,x_0\oplus x_1, x_0\oplus x_0\cdot x_1,x_1\oplus x_0\cdot x_1,x_0\oplus x_1\oplus x_0\cdot x_1},
\end{equation}
because we can always add a constant term to the random variable $g$ for a fixed $y$.
More generally, we can represent all symmetry transformations of Alice's input as
\begin{equation}\label{eq:affine}
    \begin{pmatrix}
        x_0\\
        x_1
    \end{pmatrix}\to M\begin{pmatrix}
        x_0\\
        x_1
    \end{pmatrix}+\begin{pmatrix}
        c_0\\
        c_1
    \end{pmatrix},
\end{equation}
for invertible $M\in \Set{0,1}^{2\times 2}$ and $c_0,c_1\in \{0,1\}$, which captures the relabelling of the bits $x_0\leftrightarrow x_1$, adding a constant term to each of them and substituting one of the bits with the XOR of the two bits, e.g., $x_0\to x_0\oplus x_1$.

Our goal is to classify all expressions of the form
\begin{multline}\label{eq:ULB_class}
    I(g;f_{00}(x_0,x_1)\vert y=00)+I(g;f_{01}(x_0,x_1)\vert f_{00}(x_0,x_1),y=01)+I(g;f_{10}(x_0,x_1)\vert f_{00}(x_0,x_1),f_{01}(x_0,x_1),y=10)\\+I(g;f_{11}(x_0,x_1)\vert f_{00}(x_0,x_1),f_{01}(x_0,x_1),f_{10}(x_0,x_1),y=11),
\end{multline}
for the functions $f_{y_0y_1}(x_0,x_1)$ chosen from the set in \cref{eq:fy0y1set}, under the affine transformations in \cref{eq:affine}.
In particular, the choice of the first function $f_{00}$ can be restricted to the subset $\Set{0,x_0,x_0\cdot x_1}$. 
Having chosen $f_{00}$, the symmetries can similarly help us to restrict the choice of $f_{01}$, and so on.
We can depict this process with a tree graph, with each node representing the choice of function $f_{y_0y_1}$ for $y_0y_1\in\Set{0,1}^2$.
In \cref{fig:tree} we show this graph for $f_{00}(x_0,x_1) = 0$, and the other two cases can be analysed similarly.

Having enumerated all combinations $(f_{00},f_{01},f_{10},f_{11})$ with the help of the above procedure, we further identify those that lead to equivalent expressions in \cref{eq:ULB_class}.
As a result, we obtain the following eight classes of the extended \ac{IC} expressions:
\begin{equation}\label{eq:8_classes}
    \begin{tabular}{r l}
    {Class I}:& \quad $0$,\\
    {Class II}:& \quad  $I(g;x_0\vert y=00)$,\\
    {Class III}:& \quad  $I(g;x_0\cdot x_1\vert y=00)$,\\
    {Class IV}:& \quad  $I(g;x_0\vert y=00)+I(g;x_1\vert x_0, y=01)$,\\
    {Class V}:& \quad  $I(g;x_0\vert y=00)+I(g;x_0\cdot x_1\vert x_0,y=01)$,\\
    {Class VI}:& \quad  $I(g;x_0\cdot x_1\vert y=00)+I(g;x_0\vert x_0\cdot x_1, y=01)$,\\
    {Class VII}:& \quad   $I(g;x_0\vert y=00)+I(g;x_0\cdot x_1\vert x_0,y=01)+I(g;x_1\vert x_0,x_0\cdot x_1,y=10)$,\\
    {Class VIII}:& \quad  $I(g;x_0\cdot x_1\vert y=00)+I(g;x_0\vert x_0\cdot x_1,y=01)+I(g;x_1\vert x_0,x_0\cdot x_1,y=10).$
    \end{tabular}
\end{equation}
We now study how the above classes compare to each other with regard to the strength of bounding the set of quantum correlations in Bell experiments.
Obviously, Class I is a trivial class and has no implications on the set of quantum correlations.
For the rest of the classes, instead of trying to test the corresponding criteria on some correlations, such as \ac{PR} boxes, we use the following simple criteria. 
Consider a map $(x_0,x_1)\rightarrow(x_0\cdot x_1,x_1)$ which Alice can apply to her input bits. Under this transformation, Class II gets mapped to Class III, and analogously, $(x_0,x_1)\rightarrow(x_0,x_0\cdot x_1)$ maps Class IV to Class V. Since the maps are non-invertible, the constraints obtained by Classes II and IV must be at least as strong as the ones in Classes III and V, respectively. 
Using a similar line of reasoning, we rank the relative bounding power of the rest of the classes and obtain a partial hierarchy shown in \cref{fig:hierarchy}.

\begin{figure}[t!]
    \centering
    \subfloat[\label{fig:tree}]{%
    \includegraphics[height=20em]{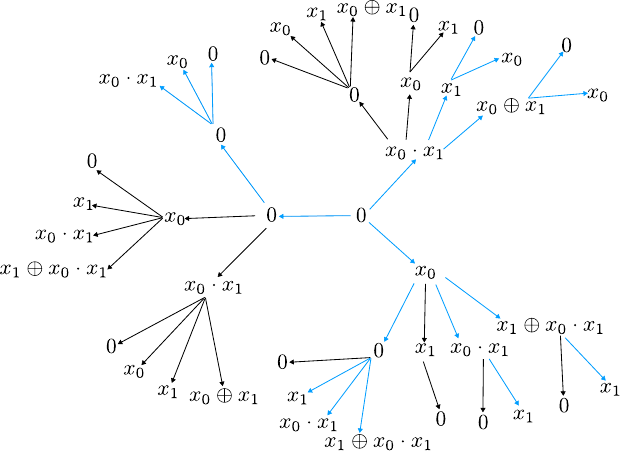}
    }\qquad
    \subfloat[\label{fig:hierarchy}]{%
    \includegraphics[height=20em]{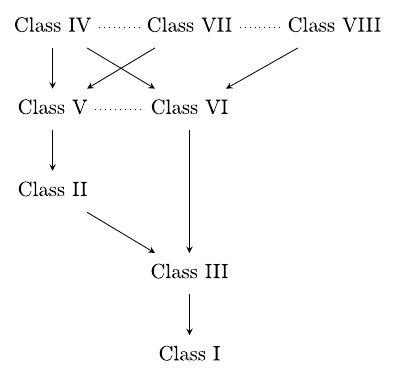}
}
\caption{a) The process of classification of Boolean functions with equivalent form of the lower bound in \cref{th:ULB} for $X=Y=\Set{0,1}^2$. The root $0$ denotes the case of $f_{00}(x_0,x_1)=0$, and each node represents a choice for $f_{y_0y_1}$, with the leafs corresponding to the choice of $f_{11}$.
The blue arrows represent twelve unique branches, which are then further reduced to eight non-equivalent extended \ac{IC} statements.
b) Relations between the eight equivalence classes in \cref{eq:8_classes} with respect to the constraints that they impose on the set of quantum correlations.
An arrow between two classes represents that one class of functions implies the other.
}
\end{figure}

\section{Comparison between functions within the same equivalence class}\label{app:convolution}
In the Discussion and Outlook section, we mention that even though two functions have the same form of the lower bound in \cref{th:ULB}, they may have different implications on the set of quantum correlations.
For example, $\INDEX_2$, $\IP_2$, and $\DISJ_2$ all belong to Class IV, but each requires $1$, $2$, and $3$ \ac{PR}-boxes, respectively, to be computed perfectly with one bit of communication, using the construction of van Dam~\cite{vandam2013implausible}. 
In fact, it can be shown that $3$ is the maximal number of \ac{PR}-boxes required for any distributed function of input length $2$.

Indeed, let $a$ and $b$, taking values in $\Set{0,1}$ be the outputs of a \ac{PR}-box shared between Alice and Bob. 
Then for any two functions $\alpha:X\to \Set{0,1}$, and $\beta:Y\to \Set{0,1}$, Alice and Bob can distributively calculate their product $\alpha(x)\cdot\beta(y)$ (logical AND) by inserting the values of these functions in the \ac{PR}-box, and by Alice sending to Bob her output $a$, because then $a\oplus b=\alpha(x)\cdot\beta(y)$~\cite{popescu1994quantum}.
We refer to this way of calculating the logical AND of two functions as the \emph{non-local AND}.

We can now use the insights from classifications of functions in \cref{app:classification} to conclude that an arbitrary function in \cref{eq:f2_gen} can be rewritten as

\begin{equation}\label{eq:func_exp}
    \begin{split}
        f(x_0x_1,y_0y_1)&=f_{00}(x_0,x_1)\oplus y_0\cdot (f_{00}(x_0,x_1)\oplus f_{10}(x_0,x_1))\oplus y_1\cdot (f_{00}(x_0,x_1)\oplus f_{01}(x_0,x_1))\\
        &\oplus y_0\cdot y_1\cdot(f_{00}(x_0,x_1)\oplus f_{01}(x_0,x_1)\oplus f_{10}(x_0,x_1)\oplus f_{11}(x_0,x_1)),
    \end{split}
\end{equation}
and thus requires at most $3$ non-local AND operations.
Notice, that for $\DISJ_2$, we have $f_{00}(x_0,x_1)=1$, $f_{10}(x_0,x_1)=1\oplus x_0$, $f_{01}(x_0,x_1)=1\oplus x_1$, and $f_{11}(x_0,x_1)=1\oplus x_0\cdot x_1$.

An expansion similar to one in~\cref{eq:func_exp} can be shown to hold for an arbitrary input length $n$ with a maximum of $2^n-1$ non-local AND operations. 
As one would expect, $\DISJ_n$ is a function that saturates this limit. 

Now, we compare the functions from the same equivalence class in terms of their utility in bounding the set of quantum correlations.
More precisely, we consider non-perfect \ac{PR}-boxes, with 
\begin{equation}
    \Pr[a\oplus b=\alpha(x)\cdot\beta(y)] \eqqcolon \frac{1+e}{2},
\end{equation}
where $e\in [-1,1]$ is the \emph{bias} of the \ac{PR}-box, and as above we denote by $\alpha:X\to \Set{0,1}$ and $\beta:Y\to\Set{0,1}$ the inputs to the box and as $a$ and $b$ the box's outputs for Alice and Bob, respectively.
Our goal is to infer a constraint that \cref{th:ULB} can impose on the bias $e$.
As an example, let us compare $\IP_2:\Set{0,1}$ and $\INDEX_2$, which belong to Class IV~\eqref{eq:8_classes}.
Let us denote Alice's input again as $x=x_0x_1$, and Bob's input as $y=y_0y_1\in\Set{0,1}^2$ for the inner product and simply as $y\in\Set{0,1}$ for the index function, respectively.
Since $\IP_2$ needs two \ac{PR}-boxes to be computed perfectly, we add an index $i\in\Set{0,1}$ to boxes' inputs and outputs, as well as to their biases.

Let Alice and Bob use the following protocols~\cite{vandam2013implausible} for computing these functions:
\begin{align}
    \IP_2(x_0x_1,y_0y_1) &= x_0\cdot y_0\oplus x_1\cdot y_1, &  \INDEX_2(x_0x_1,y) &=x_y,\\
    \alpha_i(x) = x_i,\; \beta_i(y) &= y_i,\; i\in\Set{0,1} & \alpha(x) = x_0\oplus x_1,\; \beta(y) &=y,\nonumber \\
    M =a_0\oplus a_1,\; g &= M\oplus b_0\oplus b_1, & M =a\oplus x_0,\; g &=M\oplus b,\nonumber
\end{align}
where $M$, taking values in $\Set{0,1}$, is the message that Alice sends to Bob.

The probabilities of Bob guessing the value of the function then satisfy
\begin{multline}
    \Pr[g=\INDEX_2(x,y)\vert x=k,y=l] = \Pr[a\oplus b\oplus x_0 = x_y \vert x=k,y=l] = \Pr[a\oplus b = \alpha(k)\cdot \beta(l)] = \frac{1+e}{2},
\end{multline}
for the index function, where $k\in\Set{0,1}^2$, $l\in\Set{0,1}$, and 
\begin{multline}
    \Pr[g=\IP_2(x,y)\vert x=k,y=l] = \Pr[a_0\oplus a_1\oplus b_0\oplus b_1 = x_0\cdot y_0\oplus x_1\cdot y_1 \vert x=k,y=l] \\ = \sum_{j\in\{0,1\}}\Pr[a_0\oplus b_0 = \alpha_0(x)\cdot \beta_0(y)\oplus j\vert x=k,y=l]\Pr[a_1\oplus b_1 = \alpha_1(x)\cdot \beta_1(y)\oplus j\vert x=k,y=l] = \frac{1+e_0e_1}{2},
\end{multline}
for the inner product function, where $k,l\in\Set{0,1}^2$.
Now, when we extract the implications of the extended \ac{IC} principle on $e$ and $e_0e_1$ (in analogy to the standard \ac{IC} principle, e.g., via a method described in Ref.~\cite{jain2024informationcausality}), it is clear that we obtain a stronger constraint on $e$ in comparison to the constraint \emph{per} \ac{PR}-box on $e_0$ and $e_1$.
This argument suggests that functions that require the least number of non-local AND operations to be computed result in tighter constraints on the set of quantum correlations.

\twocolumngrid
\bibliography{bibliography}

\end{document}